\documentclass[reqno,a4paper,11pt]{article}
\pdfoutput=1
\usepackage{color}

\usepackage{graphicx}
\usepackage[textwidth = 430 pt, textheight = 630 pt]{geometry}

\definecolor{MyDarkBlue}{rgb}{0.15,0.25,0.45}
\usepackage{epsfig,rotating}
\usepackage{amsmath,amssymb}
\usepackage{amsfonts}
\usepackage{mathrsfs}
\usepackage{bbm}
\usepackage[normalem]{ulem}

\usepackage{latexsym}
\usepackage{amsthm}

\usepackage[utf8x]{inputenc}

\usepackage{hyperref}
\hypersetup{
hypertexnames=false,
colorlinks=true,
citecolor=MyDarkBlue,
linkcolor=MyDarkBlue,
urlcolor=MyDarkBlue,
pdfauthor={Getachew Alemu Demessie and Christian S\"amann},
pdftitle={Higher Poincare Lemma and Integrability},
pdfsubject={hep-th math-ph}
breaklinks=true
}

\usepackage{tikz}
\usepackage{mathtools}

\let\fn\footnote
\renewcommand{\footnote}[1]{\linespread{1.1}\fn{#1}\linespread{1.29}}

\makeatletter\renewcommand{\section}{\@startsection
{section}{1}{\z@}{-3.5ex plus -1ex minus
    -.2ex}{2.3ex plus .2ex}{\bf }}
\makeatletter\renewcommand{\subsection}{\@startsection{subsection}{2}{\z@}{-3.25ex
plus -1ex minus
   -.2ex}{1.5ex plus .2ex}{\bf }}
\makeatletter\renewcommand{\subsubsection}{\@startsection{subsubsection}{3}{-2.45ex}{-3.25ex
plus -1ex minus -.2ex}{1.5ex plus .2ex}{\it }}
\renewcommand{\thesection}{\arabic{section}}
\renewcommand{\thesubsection}{\arabic{section}.\arabic{subsection}}
\renewcommand{\@seccntformat}[1]{\@nameuse{the#1}.~~}

\renewcommand{\theequation}{\thesection.\arabic{equation}}
\makeatletter \@addtoreset{equation}{section}

\usepackage[toc,page]{appendix}

\newtheorem{thm}{Theorem}[section]
\renewcommand{\thethm}{\thesection.\arabic{thm}}
\newtheorem{lemma}[thm]{Lemma}
\newtheorem{definition}[thm]{Definition}
\newtheorem{theorem}[thm]{Theorem}
\newtheorem{proposition}[thm]{Proposition}
\newtheorem{corollary}[thm]{Corollary}

\renewcommand{\appendices}{
\section*{Appendix}\label{appendices}\setcounter{subsection}{0}
\addcontentsline{toc}{section}{Appendix}
\setcounter{equation}{0}
\makeatletter
\renewcommand{\theequation}{\Alph{subsection}.\arabic{equation}}
\renewcommand{\thesubsection}{\Alph{subsection}}
\renewcommand{\thethm}{\Alph{subsection}.\arabic{thm}}
\@addtoreset{equation}{subsection}
\@addtoreset{thm}{subsection}
\makeatother
}

\def\slasha#1{\setbox0=\hbox{$#1$}#1\hskip-\wd0\hbox to\wd0{\hss\sl/\/\hss}}

\def\periodb#1{\setbox0=\hbox{$#1$}#1\hskip-\wd0\hbox to\wd0{-}}

\newcommand{\unit}{\mathbbm{1}}   			
\newcommand{\id}{\mathrm{id}}   			

\newcommand{\CA}{\mathcal{A}}    			

\newcommand{\CC}{\mathcal{C}}

\newcommand{\CCD}{\mathscr{D}}

\newcommand{\CF}{\mathcal{F}}

\newcommand{\CH}{\mathcal{H}}

\newcommand{\frg}{\mathfrak{g}}				
\newcommand{\frh}{\mathfrak{h}}				

\newcommand{\frX}{\mathfrak{X}}

\newcommand{\frl}{\mathfrak{l}}

\newcommand{\FR}{\mathbbm{R}}     			
\newcommand{\FC}{\mathbbm{C}}     			
\newcommand{\RZ}{\mathbbm{Z}}     			

\newcommand{\dd}{\mathrm{d}}     			
\newcommand{\dpar}{\partial}     			
\newcommand{\embd}{{\hookrightarrow}}     		
\newcommand{\di}{\mathrm{i}}     			

\newcommand{\sB}{\mathsf{B}}

\newcommand{\eand}{{\qquad\mbox{and}\qquad}}     		
\newcommand{\ewith}{{\qquad\mbox{with}\qquad}}
\newcommand{\efor}{{\qquad\mbox{for}\qquad}}

\newcommand{\der}[1]{\frac{\dpar}{\dpar #1}}   		

\newcommand{\sU}{\mathsf{U}}     			

\newcommand{\sG}{\mathsf{G}}

\newcommand{\sL}{\mathsf{L}}

\newcommand{\sLie}{\mathsf{Lie}}

\newcommand{\sH}{\mathsf{H}}

\newcommand{\sGL}{\mathsf{GL}}

\newcommand{\sAut}{\mathsf{Aut}}

\newcommand{\acton}{\vartriangleright}     			
\def\tyng(#1){\hbox{\tiny$\yng(#1)$}}			
\def\tyoung(#1){\hbox{\tiny$\young(#1)$}}			

\newcommand{\beq}{\begin{eqnarray}}
\newcommand{\eeq}{\end{eqnarray}}

\newcommand{\sft}{{\sf t}}

\begin{document}

\begin{titlepage}
\begin{flushright}
 EMPG--14--12
\end{flushright}
\vskip 2.0cm
\begin{center}
{\LARGE \bf Higher Poincar\'e Lemma and Integrability}
\vskip 1.5cm
{\Large Getachew Alemu Demessie and Christian S\"amann}
\setcounter{footnote}{0}
\renewcommand{\thefootnote}{\arabic{thefootnote}}
\vskip 1cm
{\em Maxwell Institute for Mathematical Sciences\\
Department of Mathematics, Heriot-Watt University\\
Colin Maclaurin Building, Riccarton, Edinburgh EH14 4AS, U.K.}\\[0.5cm]
{Email: {\ttfamily gd132@hw.ac.uk~,~c.saemann@hw.ac.uk}}
\end{center}
\vskip 1.0cm
\begin{center}
{\bf Abstract}
\end{center}
\begin{quote}
We prove the non-abelian Poincar\'e lemma in higher gauge theory in two different ways. That is, we show that every flat local connective structure is gauge trivial. The first method uses a result by Jacobowitz which states solvability conditions for differential equations of a certain type. The second method extends a proof by T.\ Voronov and yields the explicit gauge parameters connecting a flat local connective structure to the trivial one. Finally, we show how higher flatness appears as a necessary integrability condition of a linear system which featured in recently developed twistor descriptions of higher gauge theories.
\end{quote}
\end{titlepage}

\section{Introduction and results}

Higher gauge theory \cite{Baez:2004in,Sati:0801.3480,Baez:2010ya} is an interesting generalization of ordinary gauge theory that describes consistently the parallel transport of extended objects. This requires the introduction of higher form potentials, and the usual no-go theorems concerning non-abelian higher form theories are circumvented by categorifying the mathematical structures underlying ordinary gauge theory.

The need to parallel transport extended objects arises e.g.\ in string and M-theory, where point particles are replaced by one-, two- and five-dimensional objects: the strings, the M2- and M5-branes. In particular, there is a superconformal field theory in six dimensions which can be regarded as an effective description of stacks of multiple M5-branes \cite{Witten:1995zh}. Because the interactions of M5-branes are mediated by M2-branes ending on them in so-called self-dual strings, the theory should also capture the parallel transport of these strings. This fits the fact that in the abelian case corresponding to a single M5-brane, its field content comprises a 2-form potential. Altogether, it is therefore reasonable to expect that this theory -- if it exists at the classical level -- is a higher gauge theory.

Many approaches towards constructing this six-dimensional superconformal field theory have been followed. Within the framework of higher gauge theory, the twistor constructions of \cite{Saemann:2012uq,Saemann:2013pca,Jurco:2014mva} seem particularly promising. Here, manifestly superconformal field equations are derived from a Penrose--Ward transform of holomorphic principal 2- and 3-bundles, which are holomorphic versions of non-abelian gerbes. There is in fact a one-to-one correspondence between gauge equivalence classes of solutions to the arising field equations and equivalence classes of the holomorphic principal 2- and 3-bundles. In proving this one-to-one correspondence, a higher Poincar\'e lemma enters, which says that flat connective structures are pure gauge. While it is unreasonable to assume that this statement is not true, we have not found it explicitly in the literature. In this paper, we provide two independent proofs of the higher Poincar\'e lemma, both for principal 2- and 3-bundles.

The difficulty in proving the higher Poincar\'e lemma is that one of the standard ways of showing the ordinary Poincar\'e lemma, the Frobenius theorem, cannot be readily extended beyond 1-form potentials. (Note that when speaking of the Poincar\'e lemma, we always refer to the statement that abelian or non-abelian flat connections are gauge equivalent to the trivial connection.) In particular, it does not seem to be clear what a higher generalization of the notion of foliation would be. We speculate about this in appendix \ref{app:HigherDistributions}, where we show how differential ideals are related to certain $L_\infty$-structures on multivector fields, but the picture remains incomplete. Fortunately, the reformulation of the Frobenius theorem as an equation in differential forms has a generalization due to Jacobowitz \cite{jacobowitz1978}. This generalization is sufficient to establish a first proof of the higher Poincar\'e lemma for principal 2- and 3-bundles. 

Another way of proving the Poincar\'e lemma has been recently followed by Voronov \cite{Voronov:0905.0287}. Here, the explicit gauge parameters connecting the flat connection to the trivial one are constructed from a Cauchy problem. We find nice generalizations of this proof to the case of principal 2- and 3-bundles. It should be noted that Voronov's proof holds for connections taking values in a Lie superalgebra or even in a $\RZ$-graded Lie algebra, see also \cite{Igusa:0912.0249} in this context. As differential graded algebras can be regarded as duals to higher Lie or $L_\infty$-algebras, Voronov's proof contains to some extent already a dual description of the Poincar\'e lemma for higher gauge theory based on $Q$-bundles and $Q$-groups. Our generalization of his proof, however, gives directly the picture in ordinary higher gauge theory.

Flat connections arise in twistor descriptions of  gauge field equations as solutions to linear systems of the form $(\dd+A)g=0$, where $g$ is a matrix group valued function and $A$ is a matrix Lie-algebra valued one-form. This linear system directly implies that $A=\dd g g^{-1}$ is pure gauge. Moreover, it can only have a solution if the curvature $F:=\dd A+\tfrac12[A,A]$ vanishes. The Frobenius theorem or, equivalently, the Poincar\'e lemma then states that this condition is in fact sufficient for the existence of a solution. It is interesting to see if and how these statements generalize to the higher case. As we show, the higher analogue of having a matrix group for crossed modules of Lie groups is to have an underlying $A_\infty$-algebra structure. If the products of this structure extend to the Lie groups, one can indeed write down a linear system containing a flat connective structure on a principal 2- or 3-bundle which implies that the connective structure is gauge equivalent to the trivial one and that the corresponding curvatures vanish. We expect that this observation has interesting applications in generalizing notions and structures from the theory of classical integrable systems to the higher setting.

This paper is organized as follows. In section \ref{sec:2}, we review Jacobowitz's theorem and use it to give a first proof of the higher Poincar\'e lemma. In section \ref{sec:3}, we show how Voronov's proof of the Poincar\'e lemma is extended to the higher situation. Finally, section \ref{sec:integrability} shows how higher flatness can be seen as a necessary integrability condition on a linear system. Appendix  \ref{app:HigherDistributions} contains some speculations relating $L_\infty$-structures on multivector fields to differential ideals.

\section{The Poincar\'e lemma for higher gauge theory}\label{sec:2}

As the Poincar\'e lemma is a local statement, we shall be merely interested in the local description of higher gauge theories. That is, we consider local connective structures on principal $n$-bundles, which are encoded in certain differential forms on an open contractible patches of a smooth manifold. We ignore all issues related to patching these local objects to global ones.

The local description of higher gauge theory is readily derived, cf.\ e.g.\ \cite{Jurco:2014mva}. Consider the tensor product of the differential graded algebra of differential forms $\Omega^\bullet(U)$ on a patch $U$ of a smooth manifold with a semistrict gauge Lie $n$-algebra in the form of an $n$-term strong homotopy Lie algebra. The result is another strong homotopy Lie algebra, whose Maurer--Cartan equations have solutions describing flat local connective structures on semistrict principal $n$-bundles over $U$. One can read off the definition of curvatures as well as the infinitesimal gauge transformations of the differential forms defining the connective structure. To derive the finite gauge transformations, however, one has to work a little harder.

We shall restrict our discussion to the case of principal 2- and 3-bundles with strict gauge 2- and 3-groups. It is hard to imagine that an analogous statement fails to hold in the semistrict case or for higher principal $n$-bundles, and a proof for these cases along similar lines to the ones below should exist. This, however, is not obvious. Moreover, one might want to use a different set up for such a proof, as e.g.\ encoding higher gauge groups in simplicial manifolds.

\subsection{A generalized Poincar\'e lemma}

The usual Poincar\'e lemma states that the equation $\dd \alpha=\beta$ involving some $p$- and $p+1$-forms $\alpha$ and $\beta$ can be solved in an open, contractible region if and only if $\dd \beta=0$. In \cite{jacobowitz1978}, Jacobowitz presented a generalization of this statement which we briefly review below. The precise definition of having local solutions is as follows.
\begin{definition}\label{def:solvable}
 We say that the equation $\dd \omega=\Psi_{p+1}(x,\omega)$ for a $p$-form $\omega$ is \uline{solvable in a region $D$}, if for each $x\in D$ and for each $\omega_0\in \wedge^p T^{*}M|_x$, there is an open neighborhood $U_x\subset D$ and an $\omega\in \Omega^p(U_x)$ such that $\dd \omega=\Psi_{p+1}(x,\omega)$ and $\omega|_x=\omega_0$.
\end{definition}
\noindent The generalized Poincar\'e lemma reads then as follows.
\begin{proposition}\label{prop:Poincare-lemma}
 The equation $\dd \omega=\Psi_{p+1}(x,\omega)$ is solvable in a region $D$, if for all $x\in D$ there is a neighborhood $U_x$ such that for all $\omega_0\in \Omega^p(U_x)$ with $\dd \omega_0=\Psi_{p+1}(\omega_0)$ at $x$, we have $\dd \Psi_{p+1}(\omega_0)=0$ at $x$. This statement generalizes to systems of such equations with forms $\omega$ of varying degree.
\end{proposition}
\noindent The proof found in \cite{jacobowitz1978} is a generalization of the usual proof of the Frobenius theorem.

Recall that the ordinary Frobenius theorem states that an involutive distribution $\CCD$ on a manifold $M$ (i.e.\ a smoothly varying family of subspaces of the tangent bundle, on whose sections the Lie bracket of vector fields closes) corresponds to a regular foliation of $M$ by submanifolds $N$. In modern language, the distribution is the annihilator of a differential ideal generated by 1-forms. Such a differential ideal comes with integral submanifolds. That is, for each point $p\in M$, we have an embedding $e:N_p\embd M$ such that $p\in N_p$ and  $e^*\alpha=0$ for any form $\alpha$ in the differential ideal. These integral submanifolds correspond to the leaves of the foliation of $M$.

It does not seem to be completely clear how to generalize this picture to higher forms. The equation $\dd \omega=\Psi_{p+1}(x,\omega)$ is certainly again encoded in a differential ideal which, however, is no longer generated exclusively by 1-forms. Such an ideal forms an exterior differential system, which admits integral submanifolds if and only if Cartan's test is passed, cf.\ \cite{MR1083148}. One issue with Cartan's test is that it does not work in the smooth, but only in the real analytic category. In appendix \ref{app:HigherDistributions}, we present some partial generalization of the notion of distribution, which amounts to a differential ideal. The conditions of Cartan's test, however, do not seem to have a clear interpretation in the context of generalized distributions.

\subsection{Local flat connective structures on principal 2-bundles}

A principal 2-bundle is essentially the non-abelian generalization of a gerbe, see \cite{Breen:math0106083,Aschieri:2003mw,Bartels:2004aa}. Connections on principal 2-bundles were discussed in detail in \cite{Baez:2004in}. Here, we will only need the local description over an open, contractible patch $U$ of a smooth manifold $M$ and the only non-trivial data will be the local connective structure over the patch $U$.

Principal 2-bundles come with a structure Lie 2-group. The most general Lie 2-groups are notoriously difficult to handle, and we therefore restrict our attention in this paper to strict such 2-groups. These are well-known to be equivalent to crossed modules of Lie groups, cf.\ \cite{Baez:0307200}.
\begin{definition}\label{def:crossed_module}
A \uline{crossed module of Lie groups} $(\sH\xrightarrow{~\sft~}\sG,\acton)$ is a pair of Lie groups $\sG$ and $\sH$ together with a group homomorphism $\sft:\sH\rightarrow \sG$ and an action by automorphism $\acton$ of $\sG$ on $\sH$. The group homomorphism and the action satisfy the following compatibility conditions for all $g\in\sG$ and $h,h_{1,2}\in\sH$:
\begin{equation}
  \sft(g\acton h)\ =\ g \sft(h) g^{-1}\eand \sft(h_1)\acton h_2\ =\ h_1h_2h_1^{-1}~.
\end{equation}
The first condition guarantees equivariance with respect to conjugation, while the second condition is the \uline{Peiffer identity}.
\end{definition}
Applying the tangent functor to a crossed module of Lie groups, we obtain the following.
\begin{definition}
 A \uline{crossed module of Lie algebras} $(\frh\xrightarrow{~\sft~}\frg,\acton)$ is a pair of Lie algebras $\frg$ and $\frh$ together with a Lie algebra homomorphism $\sft:\frh\rightarrow \frg$ and an action by derivation $\acton$ of $\frg$ on $\frh$. The compatibility conditions here read as:
\begin{equation}
 \sft(\gamma\acton \chi)\ =\ [\gamma,\sft(\chi)]\eand \sft(\chi_1)\acton \chi_2\ =\ [\chi_1,\chi_2]
\end{equation}
for all $\gamma\in \frg$ and $\chi, \chi_{1,2}\in \frh$.
\end{definition}

The standard example of a crossed module of Lie groups is the automorphism 2-group $(\sG\xrightarrow{~\sft~}\sAut(\sG),\acton)$ of a Lie group $\sG$, where $\sft$ is the embedding by the adjoint action and $\acton$ is the automorphism action. Another example is the delooping $\sB\sU(1):=(\sU(1)\xrightarrow{~\sft~}*,\acton)$ of $\sU(1)$, where $*=\{\unit\}$ is the trivial group  and $\sft$ and $\acton$ are trivial.

Instead of delving into the general definition of principal 2-bundles, we merely need the local description of their connective structures.
\begin{definition}
 Given an open, contractible patch $U$ of a smooth manifold $M$, a \uline{local connective structure} over $U$ of a principal 2-bundle with structure crossed module $(\sH\xrightarrow{~\sft~}\sG,\acton)$ is given by a $\sLie(\sG)$-valued 1-form $A$ together with a $\sLie(\sH)$-valued 2-form $B$ over $U$. The corresponding \uline{curvatures} read as 
 \begin{equation}
  \CF:=\dd A+\tfrac{1}{2}[A,A]-\sft(B)\eand H:=\dd B+A\acton B~.
 \end{equation}
 An equivalence relation on local connective structures is given by \uline{gauge transformations}, which are parameterized by a $\sG$-valued function $g$ together with a $\sLie(\sH)$-valued 1-form $\Lambda$ as follows:
  \begin{equation}\label{eq:gauge-trafos-2}
    \begin{aligned}
    A\ &\mapsto\ \tilde{A}\ :=\ g^{-1} A g+g^{-1} \dd g-\sft(\Lambda)~,\\
    B\ &\mapsto\ \tilde{B}\ :=\ g^{-1}\acton B -\dd \Lambda-\tilde{A}\acton\Lambda-\tfrac12[\Lambda,\Lambda]~,\\
    \CF\ &\mapsto\ \tilde{\CF}\ \,:=\ g^{-1} \CF g~,\\
    H\ &\mapsto\ \tilde{H}\ :=\ g^{-1}\acton H-\tilde \CF\acton \Lambda~.
    \end{aligned}
  \end{equation}
\end{definition}
If a connective structure is to describe a consistent parallel transport of a 1-dimensional object along a surface, the curvature, also called ``fake curvature", $\CF$ has to vanish. Note that the equation $\CF=0$ is invariant under gauge transformations \eqref{eq:gauge-trafos-2}.
\begin{definition}
 We call a local connective structure $(A,B)$ \uline{flat}, if $\CF=0$ and $H=0$.
\end{definition}
\noindent Note that a flat connective structure remains flat under gauge transformations \eqref{eq:gauge-trafos-2}.

We now have the following statement about flat connective structures:
\begin{theorem}\label{thm:Poincare-2}
For any flat local connective structure $(A,B)$ on a patch $U$ and any point $p\in U$, there is a neighborhood $U_p$ of $p$ such that $(A,B)$ is pure gauge.
 That is, it can be written as
 \begin{equation}\label{eq:pure-gauge-2}
  \begin{aligned}
    A\ & =\ g^{-1} \dd g-\sft(\Lambda)~,\\
    B\ & =\ -\dd \Lambda-A\acton\Lambda-\tfrac12[\Lambda,\Lambda]~
  \end{aligned}
 \end{equation}
 for some $\sG$-valued function $g$ and $\frh$-valued 1-form $\Lambda$ on $U_p\subset U$.
\end{theorem}
\begin{proof}
 For simplicity, we assume that $\sG$ and $\sH$ are matrix groups. The proof is, however, readily extended to the general case. We can rewrite equations \eqref{eq:pure-gauge-2} as 
 \begin{equation}\label{eq:pure-gauge-2aa}
 \begin{aligned}
  \dd g^{-1}&=-A g^{-1}-\sft(\Lambda)g^{-1}=:\Psi_1(g,\Lambda)~,\\
  \dd \Lambda&=-B-A\acton\Lambda-\tfrac12[\Lambda,\Lambda]=:\Psi_2(g,\Lambda)~.
 \end{aligned}
 \end{equation}
 We regard \eqref{eq:pure-gauge-2aa} as a system of equations of the form $\dd \omega = \Psi_{p+1}(\omega,x)$ with $\dim(\sG)$ $0$-forms and $\dim(\sLie(\sH))$ $1$-forms. To apply proposition \ref{prop:Poincare-lemma}, we merely have to show that $\dd \Psi_1(g_0,\Lambda_0)=0$ and $\dd \Psi_2(g_0,\Lambda_0)=0$ at any $x\in U$ if $\CF=H=0$ as well as $\dd g_0^{-1}=\Psi_1(g_0,\Lambda_0)$ and $\dd \Lambda_0=\Psi_2(g_0,\Lambda_0)$ at $x$. We compute
 \begin{subequations}
 \begin{equation}
 \begin{aligned}
  \left.\dd \Psi_1(g_0,\Lambda_0)\right|_{x}&=\left.\left(-\dd A g^{-1}_0+A\wedge \dd g^{-1}_0-\sft(\dd \Lambda_0)g^{-1}_0+\sft(\Lambda_0)\wedge \dd g^{-1}_0\right)\right|_{x}\\
  &=\left(A\wedge A g^{-1}_0-\sft(B)g^{-1}_0+(A+\sft(\Lambda_0))\wedge \Psi_1(g_0,\Lambda_0)-\sft(\Psi_2(g_0,\Lambda_0))g^{-1}_0)\right|_{x}\\
  &=0 
\end{aligned}
\end{equation}
and
\begin{equation}
\begin{aligned}
  \left.\dd \Psi_2(g_0,\Lambda_0)\right|_{x}&=\left.\left(-\dd B-\dd A\acton \Lambda_0+ A\acton \dd \Lambda_0-[\dd \Lambda_0,\Lambda_0]\right)\right|_{x}\\
  &=\left.\left(A\acton B+(A\wedge A-\sft(B))\acton \Lambda_0+(A+\sft(\Lambda_0))\acton\Psi_2(g_0,\Lambda_0)\right)\right|_{x}\\
  &=0~.
 \end{aligned}
 \end{equation}
 \end{subequations}
 Therefore by proposition \ref{prop:Poincare-lemma}, equations \eqref{eq:pure-gauge-2aa}, and thus also \eqref{eq:pure-gauge-2}, are solvable on $U$. According to definition \ref{def:solvable}, this means that there is a solution in a neighborhood $U_p\subset U$ of each point $p\in U$.
\end{proof}

\subsection{Local flat connective structures on principal 3-bundles}

In this section we extend the result of the previous section to local connective structures on principal $3$-bundles. Principal 3-bundles are one step further in the categorification of principal bundles, and form non-abelian generalizations of 2-gerbes. The full description of principal 3-bundles with connective structure is found in \cite{Saemann:2013pca}, see also \cite{Martins:2009aa,Jurco:2009px} for partial earlier accounts.

Principal $3$-bundles use Lie $3$-groups as structure 3-groups, and we shall restrict ourselves to semistrict 3-groups for simplicity. Just as strict Lie 2-groups are categorically equivalent to crossed modules of Lie 2-groups, semistrict Lie 3-groups are equivalent to 2-crossed modules of Lie groups. We therefore start by recalling the latter notion \cite{Conduche:1984:155}.

\begin{definition}\label{def:2-crossed_module}
  A \uline{$2$-crossed module of Lie groups} is a normal complex of Lie groups (i.e.\ a complex of Lie groups in which each image of $\sft$ is a normal subgroup of the next group)
  \begin{equation}\label{eq:2-crossed-sequence}
  \sL \xrightarrow{~~\sft~~}\sH\xrightarrow{~~\sft~~}\sG~, 
  \end{equation}
  together with an action, $\acton$,  of $\sG$ on $\sH$ and $\sL$ by automorphism  as well as a $\sG$-equivariant binary map  $\{\cdot, \cdot\}: \sH\times \sH\longrightarrow \sL$ satisfying the following  conditions.
  For all $h, h_{1,2,3} \in H$, $g\in \sG$ and $\ell, \ell_{1,2} \in \sL$, we have
  \begin{itemize}
  \item[(i)] $\sft\left(g\acton \ell \right)= g\acton\sft(\ell)$ ~and~ $\sft(g\acton h)=g\sft(h)g^{-1}$,
  \item[(ii)] $\sft \left(\{h_1, h_2\} \right)= \left(h_1h_2h^{-1}_{1} \right)\left( \sft(h_1) \acton h_2^{-1}\right)$,
  \item[(iii)] $\{\sft(\ell_1), \sft(\ell_2) \}= \ell_1\ell_2\ell^{-1}_{1}\ell^{-1}_{2}:=[\ell_1, \ell_2]$,
  \item[(iv)]  $\{h_1h_2, h_3\}= \{ h_1, h_2h_3h_2^{-1}\}\left( \sft(h_1)\acton \{ h_2, h_3\}\right)$,
  \item[(v)] $ \{h_1, h_2h_3 \}=\{ h_1,h_2 \} \{h_1,h_3 \} \{ \sft\left(\{h_1, h_3\}\right)^{-1}, \sft(h_1)\acton h_2 \}$,
  \item[(vi)] $\{h, \sft(\ell) \}=\left( \{\sft(\ell), h \}\right)^{-1}\ell\left( \sft(h)\acton \ell^{-1}\right)$.
  \end{itemize}
\end{definition}
The map $\{\cdot, \cdot\}$ is  called the Peiffer lifting and measures the failure  of $(\sH\xrightarrow{t}\sG, \acton)$ to be a crossed module. Sometimes we use $\sL \xrightarrow{}\sH\xrightarrow{}\sG $  to denote $2$-crossed modules.  Lie $2$-crossed modules are generalizations of Lie crossed modules. In particular, we can obtain Lie crossed modules from Lie $2$-crossed modules by taking $\sL$ to be the trivial Lie group. Moreover, the Lie 2-crossed module $(\sL \xrightarrow{~\sft~}\sH, \acton )$ together with the induced action 
\begin{equation}
 h \acton \ell := \ell \{ \sft(\ell)^{-1}, h\}  
\end{equation}
for all $h\in \sH$ and $\ell \in \sL$ also forms a Lie crossed module.

Applying the tangent functor to the normal sequence \eqref{eq:2-crossed-sequence}, we obtain the axioms for $2$-crossed modules of Lie algebras. 
\begin{definition}
Let $\left( \frl,\frh, \frg \right)$ be a triple of Lie algebras. A \uline{2-crossed module of Lie algebras} (or a differential Lie $2$-crossed module) is a normal complex\footnote{i.e.\ a complex in which the image of each term is an ideal of the next} of Lie algebras
\begin{equation}
 \frl \xrightarrow{~~\sft~~}\frh \xrightarrow{~~\sft~~}\frg~,
\end{equation}
together with actions $\acton$ of $\frg$ on $\frl$ and $\frh$ by derivation  as well as a $\frg$-equivariant bilinear map,  $\{\cdot, \cdot\}: \frh \times \frh \longrightarrow \frl$ satisfying the conditions
\begin{itemize}
\item[(i)] $\sft(\gamma \acton \lambda)= \gamma \acton \sft( \lambda)  $ and $\sft(\gamma \acton \chi) = [\gamma, \sft(\chi) ]$,
\item[(ii)] $ \sft\left( \{\chi_1, \chi_2 \}\right)= [\chi_1, \chi_2]~-~\sft(\chi_1)\acton \chi_2$,
\item[(iii)] $ \{ \sft(\lambda_1), \sft(\lambda_2)\}=[\lambda_1, \lambda_2]$,
\item[(iv)] $ \{[\chi_1, \chi_2], \chi_3\}=\sft(\chi_1)\acton \{\chi_2, \chi_3\}~+~ \{\chi_1, [\chi_2, \chi_3]\}~-~\sft(\chi_2)\acton \{\chi_1, \chi_3\}~-~\{\chi_2, [\chi_1, \chi_3] \}$,
\item[(v)] $ \{\chi_1, [\chi_2, \chi_2]\}=\{\sft\left(\{\chi_1, \chi_2\} \right), \chi_3\}~-~\{\sft\left(\{\chi_1, \chi_3\} \right),\chi_2\}$,
\item[(vi)] $ -\{\sft(\lambda), \chi\} = \{\chi, \sft(\lambda)\}~+~ \sft(\chi) \acton \lambda$,
\end{itemize}
for every $\gamma \in \frg$, $\chi, \chi_{1,2,3}\in \frh$, and $\lambda, \lambda_{1,2} \in \frl$.
\end{definition}

Note that a Lie 2-crossed module of Lie groups can be partially linearized to obtain more general actions, as e.g.\ the action of $\sG$ onto $\frh$. More details on 2-crossed modules can be found in \cite{Martins:2009aa,Saemann:2013pca}.

The local description of a connective structure on a principal 3-bundle is now readily given, cf.\ \cite{Saemann:2013pca}.
\begin{definition}
Let $U$ be a contractible patch of a smooth manifold $M$. A local connective structure over $U$ of a principal $3$-bundle with structure $2$-crossed module $\Big(\sL \xrightarrow{~} \sH \xrightarrow{~} \sG,$ $\acton, \{\cdot, \cdot \}\Big)$ can be expressed as a triple of Lie algebra valued forms $(A, B, C)$, where $A\in \Omega^1\left(U, \sLie(\sG) \right)$, $B \in \Omega^2\left(U, \sLie(\sH) \right)$ and $C \in \Omega^3\left(U, \sLie(\sL) \right)$.  Corresponding curvatures are defined according to
\begin{equation}
\CF:= \dd A + \frac{1}{2}[A, A]-\sft( B)~,~~\CH:= \dd B+A\acton B -\sft(C)~, ~~G:= \dd C+ A\acton C + \{B, B\}~. 
\end{equation}
Gauge transformations act on the Lie algebra valued forms according to
\begin{equation}\label{eq:gauge-trafos-3}
  \begin{aligned}
  \ A\ &\mapsto\ \tilde{A}\ :=\ g^{-1} A g+g^{-1} \dd g-\sft(\Lambda)~,\\
  \  B\ &\mapsto\ \tilde{B}\ :=\ g^{-1}\acton B -(\dd +\tilde{A}\acton)\Lambda -\tfrac{1}{2} \sft(\Lambda)\acton \Lambda -\sft( \Sigma)~,\\
   \ C\ & \mapsto\ \tilde{C}:=\ g^{-1}\acton C - \left( (\dd + \tilde{A}\acton )+ \sft(\Lambda)\acton \right)\Sigma +\{\tilde{B}+ \tfrac{1}{2}(\dd + \tilde{A}\acton )\Lambda + \tfrac{1}{2}[\Lambda, \Lambda], ~\Lambda\} +\\ &~~~~~~~~~~~~ ~~~~~~~~~~~+ \{\Lambda,  \tilde{B}- \tfrac{1}{2}(\dd + \tilde{A}\acton )\Lambda- \tfrac{1}{2}[\Lambda,\Lambda] \}~,\\
   %
  \  \CF\ &\mapsto\ \tilde{\CF}\ :=\ g^{-1} \CF g~,\\
  \  \CH\ &\mapsto\ \tilde{\CH}\ :=\ g^{-1}\acton \CH-\tilde{\CF}\acton \Lambda~,\\
  \  G\ & \mapsto\ \tilde{G} \ :=  \ g^{-1}\acton G - \left( \tilde{\CF} \acton (\Sigma -\tfrac{1}{2}\{\Lambda, \Lambda \})\right)+ \{\Lambda, \tilde{\CH} \}~-~\{\tilde{\CH},\Lambda \}~-~\{\Lambda, \tilde{\CF}\acton \Lambda\}~,
  \end{aligned}
  \end{equation}
  where $g$ is a $\sG$-valued function and $\Lambda$ and $\Sigma $ are $\sLie(\sH)$ and $\sLie(\sL)$-valued $1$- and $2$-forms, respectively.
\end{definition} 

For consistency of the parallel transport described by this local connective structure, it is necessary that both the 2- and 3-form fake curvatures $\CF$ and $\CH$ vanish. 
 \begin{definition}
 A local connective structure $(A, B, C)$ is said to be \uline{flat}, if all curvatures vanish: $\CF=0$, ~$\CH=0$ and $G=0$. 
 \end{definition}
\noindent Again, note that as in the case of principal $2$-bundles, flat connective structures on principal $3$-bundles remain flat under the gauge transformations \eqref{eq:gauge-trafos-3}.

The Poincar\'e lemma here reads as follows.
\begin{theorem}\label{thm:Poincare-3}
For any flat local connective structure $(A,B,C)$ on a patch $U$ and any point $p\in U$, there is a neighborhood $U_p\subset U$ of $p$ such that $(A,B,C)$ is pure gauge. That is, it can be written as
 \begin{equation}\label{eq:pure-gauge-3}
  \begin{aligned}
    A\ & =\ g^{-1} \dd g-\sft(\Lambda)~,\\
    B\ & =\ -(\dd +A\acton)\Lambda -\tfrac{1}{2} \sft(\Lambda)\acton \Lambda -\sft( \Sigma)~,\\
    C\ & =\ - \left( (\dd + A\acton )+ \sft(\Lambda)\acton \right)\Sigma +\{B+ \tfrac{1}{2}(\dd + A\acton )\Lambda + \tfrac{1}{2}[\Lambda, \Lambda], ~\Lambda\} +\\ &~~~~~~~~~~~~~~~~~~+\{\Lambda,  B- \tfrac{1}{2}(\dd + A\acton )\Lambda- \tfrac{1}{2}[\Lambda,\Lambda] \}~,\\
  \end{aligned}
 \end{equation}
 for some $\sG$-valued function $g$, $\sLie(\sH)$-valued 1-form $\Lambda$ and $\sLie(\sL)$-valued 2-form $\Sigma$ on $U_p$.
\end{theorem}
\begin{proof}
 The proof is fully analogous to that of theorem \ref{thm:Poincare-2}, but considerably more involved. We therefore only outline the computations. First, we rewrite \eqref{eq:pure-gauge-3} as follows.
 \begin{equation}\label{eq:eds}
 \begin{aligned}
  \dd g^{-1}&=-A g^{-1}-\sft(\Lambda)g^{-1}=:\Psi_1(g,\Lambda,\Sigma)~,\\
  \dd \Lambda&=-B-A\acton\Lambda-\tfrac12\sft(\Lambda)\acton\Lambda-\sft(\Sigma)=:\Psi_2(g,\Lambda,\Sigma)~,\\
  \dd \Sigma&=-C-A\acton \Sigma-\sft(\Lambda)\acton\Sigma+\{B+ \tfrac{1}{2}(\dd + A\acton )\Lambda + \tfrac{1}{2}[\Lambda, \Lambda], ~\Lambda\} +\\ &\hspace{4.5cm}+\{\Lambda,  B- \tfrac{1}{2}(\dd + A\acton )\Lambda- \tfrac{1}{2}[\Lambda,\Lambda] \}=:\Psi_3(g,\Lambda,\Sigma)~.
 \end{aligned}
 \end{equation}
 Proposition \ref{prop:Poincare-lemma} guarantees that \eqref{eq:eds} are solvable on $U$, if $\dd \Psi_{1,2,3}(g_0,\Lambda_0,\Sigma_0)$ vanish at any $x\in U$ if $\CF=\CH=G=0$ as well as 
 \begin{equation}\label{eq:ext-derivaties-at-x}
 \dd g_0^{-1}|_x=\Psi_1(g_0,\Lambda_0,\Sigma_0)|_x~,~~~\dd \Lambda_0|_x=\Psi_2(g_0,\Lambda_0,\Sigma_0)|_x~,~~~\dd\Sigma_0|_x=\Psi_3(g_0,\Lambda_0,\Sigma_0)|_x~.
 \end{equation}
We now have to rewrite $\dd \Psi_{1,2,3}(g_0,\Lambda_0,\Sigma_0)$ in terms of quantities which we know at $x$. The exterior derivative will hit either a potential $n$-form or a gauge parameter. The exterior derivatives of the gauge parameters are given in \eqref{eq:eds} and the exterior derivatives of the potential $n$-forms can be rewritten using the flatness equations $\CF=\CH=G=0$. 

Putting everything together, we find after a lengthy calculation that given \eqref{eq:ext-derivaties-at-x}, $\dd \Psi_{1,2,3}(g_0,\Lambda_0,\Sigma_0)$ indeed vanish for flat local connective structures. Again, solvability of \eqref{eq:eds} over $U$ implies that for all $p\in U$ there exists a neighborhood $U_p$ over which \eqref{eq:eds} have a solution.
\end{proof}

\section{Constructive proof of the Poincar\'e lemma}\label{sec:3}

We come now to a constructive proof of the Poincar\'e lemma which yields the explicit gauge transformation trivializing a flat local connective structure. Our proof will be a direct generalization of that of \cite{Voronov:0905.0287}, where the author constructs a solution to a Cauchy problem, relating pullbacks of flat connections along homotopic maps by gauge transformation. On a contractible patch of a smooth manifold, flat connections are therefore gauge equivalent to pullbacks along constant maps. This implies that flat connections are locally pure gauge. 

\subsection{Poincar\'e lemma on principal 2-bundles}

Let $U$ be an open, contractible patch of a smooth manifold $M$. Over $U\times [0,1]$, let $(\hat A,\hat B)$ be a local connective structure with underlying crossed module of Lie groups $(\sH\xrightarrow{~\sft~}\sG,\acton)$. Let $(\frh\xrightarrow{~\sft~}\frg,\acton)$ denote the corresponding crossed module of Lie algebras. To simplify our notation, we assume that $\sG$ and $\sH$ are matrix groups. We decompose the differential forms $\hat A$ and $\hat B$ according to 
\begin{equation}
 \hat A=\hat A_x+ \dd t\,\hat A_t \eand \hat B=\hat B_x+\dd t\,\hat B_t ~,
\end{equation}
where $\der{t}\lrcorner \hat A _x=0$ and $\der{t}\lrcorner \hat B_x=0$. Similarly, we decompose the exterior derivative
\begin{equation}
 \dd \omega=\dd_x \omega+\dd t~\der{t}\omega=\dd_x \omega +\dd t~\dot{\omega}~.
\end{equation}

We are interested in solutions $g\in \CC^\infty(U\times [0,1],\sG)$ and $\Lambda\in \Omega^1(U\times [0,1],\frh)$ to the following Cauchy problem, which arises by considering gauge transformations of the components $\hat{A}_t$ and $\hat{B}_t$ to $0$, cf.\ \eqref{eq:gauge-trafos-2}:
\begin{subequations}\label{eq:Cauchy-Problem-2}
\begin{equation}\label{eq:Cauchy-Problem-a}
 \dot g=-\hat A_t\,g+ g\sft(\Lambda_t)\eand \dot \Lambda_x=g^{-1}\acton \hat B_t+\dd_x\Lambda_t+\left(g^{-1}A_xg +g^{-1}\dd_xg\right)\acton \Lambda_t
\end{equation}
with initial conditions
\begin{equation}\label{initial:cond}
 g(x,0)=\unit_\sG\eand \Lambda(x,0)=0\efor x\in U~.
\end{equation}
\end{subequations}

\begin{proposition}\label{prop:3.1}
 Let $(g,\Lambda)$ be a solution to the Cauchy problem \eqref{eq:Cauchy-Problem-2}. Then
 \begin{equation}\label{eq:prop1}
  -g_1^{-1}\dd g_1+ \int_0^1\dd t~\der{t}\lrcorner(g^{-1}\hat\CF g)=g_1^{-1}\left.\hat A_x\right|_{t=1}g_1-\left.\hat A_x\right|_{t=0}- \left. \sft(\Lambda_x)\right|_{t=1}~, 
 \end{equation}
 where $g_1:=g(x,1)$ and $\hat \CF$ is the fake curvature of the local connective structure $(\hat A,\hat B)$.
\end{proposition}
\begin{proof}
 First, using \eqref{eq:Cauchy-Problem-a}, we readily compute 
 \begin{equation}\label{3bundle: I1}
  \der{t}\left(g^{-1}\dd_x g\right)=-g^{-1} (\dd_x \hat{A}_t) g+ \sft \left(g^{-1}\dd_x g \acton \Lambda_t\right)+ \dd_x\sft(\Lambda_t)~,
 \end{equation}
 and 
\begin{equation}\label{3bundle: I2}
\begin{aligned}
\der{t}\left(g^{-1}\hat{A}_xg  \right)&=  g^{-1} \left( \dot{\hat{A}}_x + [ \hat{A}_t , \hat{A}_x ]  \right) g + \sft(g^{-1}\hat{A}_xg  \acton \Lambda _t )~.
\end{aligned}
\end{equation}

 Moreover,  
 \begin{equation}\label{eq:aux1}
  g_1^{-1}\dd g_1=\left.(g^{-1}\dd_x\,g)\right|_{t=1}\eand  \left.\dd_x g\right|_{t=0}=0~.
 \end{equation}
 We would now like to rewrite \eqref{3bundle: I1} and \eqref{3bundle: I2} in terms of the fake curvature of $(\hat A,\hat B)$. Note that
 \begin{equation}
 \begin{aligned}
 \int_0^1 \dd t~\der{t}\lrcorner (g^{-1}\hat\CF g) &=\int_0^1 \dd t~g^{-1}\Big(-\dd_x\hat A_t+\dot{\hat A}_x+[\hat A_t,\hat A_x]-\sft(\hat B_t)\Big)g\\
 &= \int_0^1 \dd t~ \der{t}\left( g^{-1}\dd_x g\right) + \int_0^1 \dd t~\der{t}\left(g^{-1}\hat{A}_xg  \right)+ \\  &~~~- \int_0^1 \dd t~\sft\Big( g^{-1}\acton \hat{B}_t+\dd_x \Lambda_t+ \left(g^{-1}A_xg +g^{-1}\dd_xg\right)\acton \Lambda_t \Big)~.
 \end{aligned}
 \end{equation}
 Using \eqref{eq:Cauchy-Problem-a} and \eqref{eq:aux1}, we can further simplify this to
 \begin{equation}
  \int_0^1 \dd t~\der{t}\lrcorner (g^{-1}\hat\CF g)=g_1^{-1}\dd g_1+\int_0^1 \dd t~\der{t}(g^{-1}\hat A_xg)-\int_0^1 \dd t~\der{t}\sft(\Lambda_x)~,
 \end{equation}
 which is obviously equivalent to \eqref{eq:prop1}.
\end{proof}
\noindent Next, we prove an analogous statement involving the 3-form curvature $\hat H$ of $(\hat A,\hat B)$:
\begin{proposition}\label{prop:3.2}
 Let $(g,\Lambda)$ be a solution to the Cauchy problem \eqref{eq:Cauchy-Problem-2}. Then
 \begin{equation}\label{eq:prop2}
 \begin{aligned}
 \dd_x \Lambda_1 &+g^{-1}_{1}\dd g_{1}\acton \left.\Lambda_x\right|_{t=1}  - \left.\left( \Lambda_x \wedge \Lambda_x\right)\right|_{t=1} + \left( g^{-1}_{1}\left. \hat{A}_x\right|_{t=1}g_1 \right) \acton \left.\Lambda_x \right|_{t=1} =~\\
 & -\int_0^1\dd t~\der{t}\lrcorner\left(g^{-1}\acton \hat H-(g^{-1}\hat \CF g)\acton \Lambda \right)+\left.g_1^{-1}\acton\hat B_x\right|_{t=1}-\left.\hat B_x\right|_{t=0}~,
 \end{aligned}
  \end{equation}
 where $g_1:=g(x,1)$, $\Lambda_1=\Lambda(x,1)$ and $\hat \CF$ and $\hat H$ are the fake and 3-form curvatures of the local connective structure $(\hat A,\hat B)$.
\end{proposition}
\begin{proof}
 In this case we have
 \begin{align}\label{ddlammbda}
 \left.\dd_x \Lambda_x\right|_{t=0}= 0 \eand \dd_x \Lambda_1=\left.\dd_x \Lambda_x \right|_{t=1}~.
 \end{align}
Moreover, by direct differentiation and using \eqref{eq:Cauchy-Problem-a}, we obtain
  \begin{align}\label{dlammbda}
  \der{t}\left( d_x\Lambda_x \right) =\dd_x\left( g^{-1}\acton \hat{B}_t + \dd_x\Lambda_t ~+ (g^{-1}\hat{A}_xg + g^{-1}\dd_xg)\acton \Lambda_t \right)~,
\end{align}
and 
\begin{equation}\label{dxg}
  \begin{aligned}
 \dfrac{\partial}{\partial t}\left( g^{-1}\dd_x g \acton \Lambda_x \right)&= \left( -g^{-1} (\dd_x \hat{A}_t) g+ \sft(g^{-1}\dd_x g \acton \Lambda_t) + \dd_x\sft(\Lambda_t) \right) \acton \Lambda_x ~+\\
&~+g^{-1}\dd_x g\acton \left( g^{-1}\acton \hat{B}_t + \dd_x\Lambda_t+ \left( g^{-1}\hat{A}_xg + g^{-1}\dd_x g\right) \acton \Lambda_t \right)~.\\
\end{aligned}
\end{equation}
      
Thus, considering the expressions of the fake and the $3$-curvatures of a local connective structure $(\hat{A}, \hat{B})$ yields 

\begin{equation}\label{3curv:main}
\begin{aligned}
\int_{0}^{1}\dd t~\der{t}&\lrcorner \Big(- g^{-1}\hat{\CF} g\acton \Lambda + g^{-1}\acton \hat{H}\Big)= \\ &\int_{0}^{1}\dd t~ \left(  g^{-1}\dd_x \hat{A}_t g \acton \Lambda_x   - g^{-1}\left( \dot{\hat{A}}_x  + [ \hat{A}_t,\hat{A}_x ] \right)g \acton \Lambda_x \right)~ +\\ 
+ & \int_{0}^{1}\dd t~ \left(- g^{-1}\left( \dd_x \hat{A}_x + \hat{A}_x \wedge \hat{A}_x  \right) g \acton \Lambda_t  + g^{-1}\sft(\hat{B}_x ) g \acton \Lambda_t  \right)+ \\
+& \int_{0}^{1}\dd t~ \Big( g^{-1}\sft(\hat{B}_t ) g \acton \Lambda_x + g^{-1}\acton \left( \dot{\hat{B}}_x +\hat{A}_t \acton \hat{B}_x -\dd_x \hat{B}_t - \hat{A}_x \acton \hat{B}_t  \right) \Big )  ~.
\end{aligned}
\end{equation}
But by direct differentiation and using \eqref{eq:Cauchy-Problem-a} we have 
\begin{equation}\label{3curv:III}
\begin{aligned}
\der{t}\Big( (g^{-1}\hat{A}_xg )&\acton \Lambda_x \Big)= \left( g^{-1} \left( \dot{\hat{A}}_x + [ \hat{A}_t , \hat{A}_x ] \right) g  \right)\acton \Lambda_x +  \sft(g^{-1}\hat{A}_xg  \acton \Lambda _t ) \acton \Lambda_x ~+ \\
~~~&+ ( g^{-1}\hat{A}_x g ) \acton \left( g^{-1} \acton \hat{B}_t + \dd_x \Lambda_t +  g^{-1}\hat{A}_xg \acton \Lambda_t  +  g^{-1}\dd_x g \acton \Lambda_t  \right)~, \\
\end{aligned}
\end{equation}
and
\begin{align}\label{3curv:Iv}
&\dfrac{\partial}{\partial t}\left(g^{-1}\acton \hat{B}_x \right)= \left( g^{-1}\hat{A}_t  -\sft(\Lambda_t)g^{-1}\right) \acton \hat{B}_x +g^{-1}\acton \dot{\hat{B}}_x~.
\end{align}
Now applying  \eqref{eq:Cauchy-Problem-a}, after combining \eqref{dlammbda}, \eqref{dxg}, \eqref{3curv:III} and \eqref{3curv:Iv}, gives 
\begin{equation}\label{3curv:derv}
\begin{aligned}
\int_{0}^{1}\dd t~\der{t}  \lrcorner\Big(  - g^{-1}\hat{\CF} &g\acton \Lambda + g^{-1}\acton \hat{H}  \Big)=\\
&\int_{0}^{1}\dd t~\der{t}\left( -d_x\Lambda_x  \right)   +  \int_{0}^{1}\dd t~\dfrac{\partial}{\partial t}  \left( - g^{-1}\dd_x g \acton \Lambda_x  \right)  ~+ \\
+&\int_{0}^{1}\dd t~\der{t} \left( -  (g^{-1}\hat{A}_xg )\acton \Lambda_x \right) + \int_{0}^{1}\dd t~\der{t}  \left(  g^{-1}\acton \hat{B}_x \right) ~+\\
+&\int_{0}^{1}\dd t~ \left( \sft ( \dot{\Lambda}_x )\acton \Lambda_x \right)~.
\end{aligned}
\end{equation}

After simplification of \eqref{3curv:derv} using \eqref{initial:cond} and \eqref{ddlammbda}, we finally arrive at
\begin{equation}
 \begin{aligned}
 \dd_x \Lambda_1 + g^{-1}_{1}&\dd g_{1}\acton \left.\Lambda_x\right|_{t=1}  -\left.\left( \Lambda_x \wedge \Lambda_x\right)\right|_{t=1} + \left( g^{-1}_{1}\left. \hat{A}_x\right|_{t=1}g_1 \right) \acton \left.\Lambda_x \right|_{t=1} =  \\
 & - \int_{0}^{1}\dd t~\dfrac{\partial}{\partial t} \lrcorner\left( -g^{-1}\hat{\CF} g \acton \Lambda + g^{-1}\acton \hat{H} \right) +
 g^{-1}_{1}\acton \hat{B}_x\vert_{t=1}-\hat{B}_x\vert_{t=0}~. 
 \end{aligned}
\end{equation}
\end{proof}

We can now follow \cite{Voronov:0905.0287} further and consider homotopic maps $h_{0,1}(x):U\rightrightarrows V$ between local patches $U$ and $V$ of some smooth manifolds. Let $h(x,t):U\times [0,1]\rightarrow V$ with $h(x,0)=h_0(x)$ and $h(x,1)=h_1(x)$ be a homotopy satisfying $\der{t} h(x,t)|_{t=0,1}=0$. Because the pullback is compatible with the wedge product and the exterior derivative, propositions \ref{prop:3.1} and \ref{prop:3.2} yield the following corollary.
\begin{corollary}
 The pullbacks of a local connective structure $(A,B)$ on the patch $V$ of some manifold along homotopic maps $h_{0,1}:U\rightrightarrows V$ are related as follows:
 \begin{equation}
  \begin{aligned}
  -g_1^{-1}\dd g_1+\sft(\Lambda_{1,x})+\int_0^1\dd t~\der{t}\lrcorner(g^{-1}h^*(\CF) g)&=g_1^{-1}h_1^*( A_x)g_1-h_0^* A_x~, \\
 \dd_x \Lambda_1+\left(g_1^{-1}h_1^* (A_x)g_1\right)\acton \Lambda_{1,x}+(g_1^{-1}\dd g_1)\acton \Lambda_{1,x}-\left(\Lambda_{1,x}\wedge\Lambda_{1,x}\right)&=~\\
 -\int_0^1\dd t~\der{t}\lrcorner\left(g^{-1}\acton h^*H-g^{-1}h^*(\CF) g\acton \Lambda\right)&+g_1^{-1}\acton h_1^* B_x-h_0^* B_x~,
  \end{aligned}
 \end{equation}
 where $h$ denotes a homotopy between $h_0$ and $h_1$ with $\der{t} h(x,t)|_{t=0,1}=0$, $(g,\Lambda )$ is a solution of the Cauchy problem \eqref{eq:Cauchy-Problem-2} and $g_1=g(x,1)$, $\Lambda_1=\Lambda(x,1)$. In particular, the pullbacks for flat connective structures are gauge equivalent.
\end{corollary}

This corollary can now be used to prove the Poincar\'e lemma. Consider an open contractible patch $U$ of a smooth manifold and regard it as a subset of some vector space $\FR^d$ containing the origin $0_U$. We are interested in the homotopy $h(x,t):U\times [0,1]\rightarrow U$ with $h(x,t)=xtk(t)$ between $U$ and the point $0_U\in U$, where $k(t)$ is a smooth function such that $k'(t)|_{t=0,1}=0$, $k(0)=0$ and $k(1)=1$. Note that the pullback of the connective structure on $U$ along $h_0$ vanishes, which implies the following theorem.
\begin{theorem} (Higher Poincar\'e lemma)
Flat local connective structures are gauge equivalent to the trivial connective structure.
\end{theorem}

\subsection{Poincar\'e lemma on principal 3-bundles}

An interesting aspect of our proof in the previous section was that it was not necessary to extend the interval $[0,1]$ used in the case of ordinary principal 2-bundles to $[0,1]^2$. The latter arises  if one wants to define the general transport 2-functor from the path 2-groupoid to the delooping of the strict Lie 2-group corresponding to the crossed module $\sH\rightarrow \sG$, cf.\ \cite{Schreiber:0802.0663}.

Therefore, and since all the terms in the formulas contained in our proof  have clear meanings, one can in principle readily generalize our proof to the case of local connective structures on principal 3-bundles. Let us here concisely summarize the steps. 

We start from a local connective structure $(\hat A,\hat B, \hat C)$ on $U\times [0,1]$, where $U$ is a contractible patch of some smooth manifold. Let $\sL\rightarrow \sH\rightarrow \sG$ be the relevant 2-crossed module and $\frl\rightarrow \frh\rightarrow \frg$ the corresponding linearization. The Cauchy problem is again given by equations stating that the components of the connective structures along $\dd t$ can be gauged away. Here, we have
\begin{equation}
 \begin{aligned}
  \dot g&=-\hat A_t g+g\sft(\Lambda_t)~,\\
  \dot \Lambda_x &= g^{-1}\acton \hat B_t+\dd_x\Lambda_t+(g^{-1}A_x g+g^{-1}\dd_x g)\acton \Lambda_t-\sft(\Sigma_t)~,\\
  \dot \Sigma_x &= g^{-1}\acton \hat C_t+\dots~,
 \end{aligned}
\end{equation}
where $\dots$ stands for terms easily read off from equations \eqref{eq:gauge-trafos-3}. As their explicit forms are not illuminating, we suppress them here. This will then lead to statements analogous to propositions \ref{prop:3.1} and \ref{prop:3.2}, which are of the form
\begin{equation}
 \int_0^1 \dd t \der{t}\lrcorner( \tilde{\hat{K}})=\tilde{\hat P}|_{t=1}-\hat{P}|_{t=0}~.
\end{equation}
Here, $\hat P$ is the potential $n$-form for $n=1,2,3$ and $\hat{K}$ is the corresponding curvature $n+1$-form. Furthermore, $\tilde{\hat P}$ and $\tilde{\hat K}$ denote gauge transformed objects.

These equations describe the relation between pullbacks of a local connective structure along homotopic maps. In particular, they imply that the pullbacks of flat local connective structures along homotopic maps are gauge equivalent. Considering again the homotopy $h(x,t):U\times [0,1]\rightarrow U$ with $h(x,t)=xtk(t)$ implies that flat local connective structures are pure gauge.

\section{The Poincar\'e lemma and integrability}\label{sec:integrability}

In the context of integrable systems, we often encounter linear systems of the form
\begin{equation}\label{eq:lin-system}
 \nabla g:=(\dd +A)g=0~,
\end{equation}
where $g$ is a $\sG$-valued function for some matrix Lie group $\sG$, $\dd$ is a differential and $A$ is a $\sLie(\sG)$-valued 1-form. For example, in the Penrose--Ward transform \cite{Ward:1977ta},  $\dd$ is a relative exterior derivative along a fibration and $A$ is a relative differential 1-form. Acting with $\nabla$ on \eqref{eq:lin-system}, we obtain $\nabla(\nabla g)=0$, which is equivalent to $F g:=(\nabla)^2 g=0$. Note that the product in $F g=0$ is just an ordinary matrix product. Multiplying by $g^{-1}$ from the left, we see that the existence of a solution $g$ requires that the curvature $F$ of $\nabla$ vanishes. Moreover, re-arranging equation \eqref{eq:lin-system} directly yields the relation $A=g \dd g^{-1}$, implying $F=0$.

In this section, we demonstrate how these statements translate to linear systems involving connective structures on principal 2-bundles.

\subsection{Underlying 2-term $A_\infty$- and $L_\infty$-algebras}

To write down equation \eqref{eq:lin-system}, it is crucial to have a matrix Lie group such that the expressions $\dd g$ and $A g$ make sense. Analogously, we consider a crossed module of matrix Lie groups $\sH\xrightarrow{\sft}\sG$, which yields matrix products between elements of $\frg:=\sLie(\sG)$ and $\sG$ as well as $\frh:=\sLie(\sH)$ and $\sH$. The Lie brackets are recovered by antisymmetrization of the matrix product. For our construction, we need a further product which turns into the action $\acton:\frg\times \frh\rightarrow \frh$ upon antisymmetrization. As we will explain now, the right context to look for such a product is an associative 2-term $A_\infty$-algebra.

Recall that an associative 2-term $A_\infty$-algebra is a graded vector space $\CA:=\CA_{-1}\oplus \CA_0:=\frh\oplus \frg$ together with ``products'' $m_1:\CA\rightarrow \CA$ and $m_2:\CA^{\otimes 2}\rightarrow \CA$ of degrees $1$ and $0$, respectively, such that
\begin{equation}
 m_1\circ m_1=0~,~~~m_1\circ m_2=m_2\circ(m_1\otimes \unit+\unit\otimes m_1)~,~~~m_2\circ(\unit\otimes m_2-m_2\otimes \unit)=0~.
\end{equation}
The first equation says that $m_1$ is a differential, the second equation states the compatibility between this differential and the product $m_2$ and the third equation implies that the product $m_2$ is associative. We use the usual sign convention for the maps $m_i$:
\begin{equation}
 (m_i\otimes m_j)(a_1\otimes a_2)=(-1)^{\tilde{m}_j\tilde{a}_1}m_i(a_1)\otimes m_j(a_2)~,
\end{equation}
where $\tilde{a}_1$ denotes the total parity of $a_1\in \CA^{\otimes i}$ and $\tilde{m}_j:=2-j$. 

If we antisymmetrize the products $m_i$ to antisymmetric products $\mu_i$, we obtain a 2-term $L_\infty$-algebra, cf.\ \cite{Stasheff:1963aa,Stasheff:1963ab,Stasheff:1992bb}. Associative 2-term $L_\infty$-algebras, in turn, are equivalent to crossed modules of Lie algebras. More explicitly, the map $\sft$ is identified with $m_1=\mu_1$, the commutator on $\frg:=\CA_0$ is given by $\mu_2:\frg\times \frg\rightarrow \frg$, the action of $\frg$ onto $\frh:=\CA_{-1}$ is given by $\mu_2:\frg\times \frh\rightarrow \frh$ and the commutator on $\frh$ is given by $\mu_2\circ(\mu_1\otimes \unit)$. Altogether, we conclude that the higher analogue of demanding a matrix Lie algebra structure instead of merely a Lie algebra structure implies to ask for an $A_\infty$-algebra underlying the $L_\infty$-algebra corresponding to the crossed module of Lie algebras. Finally, we demand that the $A_\infty$-product can be continued to a product between the $A_\infty$-algebra and the crossed module of Lie groups $\sH\xrightarrow{\sft}\sG$, such that we have products
\begin{equation}
 m_2:\CA_0\times \sG\rightarrow \CA_0~,~~~m_2:\CA_0\times\sH\rightarrow \CA_{-1}\eand m_2:\CA_{-1}\times \sG\rightarrow \CA_{-1}~.
\end{equation}
We now arrived at a complete higher analogue of having a matrix Lie group. 

As a non-trivial example for such a structure, consider the crossed module of Lie groups $\sH\xrightarrow{\sft}\sG=\sGL(n,\FC)\xrightarrow{\id} \sGL(n,\FC)$. The action $\acton$ is just the adjoint action, and we define
\begin{equation}
 m_2(a,b):=\left\{\begin{array}{cl}
                   ab & a,b\in \sG\cup \sLie(\sG)~\\
                   ab & a\in \sG\cup \sLie(\sG)~,~~b\in \sH\cup\sLie(\sH)~,\\
                   ab^{-1} &  a\in \sLie(\sH)~,~~b\in \sG~,\\
                   -ab & a\in \sLie(\sH)~,~~b\in \sLie(\sG)~.
                  \end{array} \right.
\end{equation}

It is not clear to us how to construct such an $A_\infty$-algebra for an arbitrary crossed module of Lie algebras $\frh\rightarrow \frg$, but we strongly suspect that there is such a construction. Even if such a construction did not exist, we could impose a restriction to crossed modules admitting such a construction. This set is not empty, as the above example shows. Note that the $A_\infty$-algebra resulting from the construction $\CA=\CA_{-1}\oplus \CA_0$ will contain the crossed module as an 2-term $L_\infty$-subalgebra and therefore might be larger than $\frg\oplus \frh$. 

To deal with connections and their curvatures, we have to allow for differential forms on some contractible region $U$ taking values in the subspace $\frh\oplus \frg$ of the $A_\infty$-algebra $\CA=\CA_{-1}\oplus \CA_0$. Recall that $\Omega^\bullet(U)$ is a differential graded algebra, and there is a natural tensor product between differential graded algebras and $A_\infty$-algebras. This product yields an $A_\infty$-algebra $\tilde \CA:=\Omega^\bullet(U)\otimes \CA$ where the total degree of an element is the sum of the degree in $\CA$ and its form degree. The products are given by
\begin{equation}
 \tilde m_1(a):=\dd a+(-1)^p m_1(a)\eand \tilde m_2=m_2
\end{equation}
for $a\in\Omega^p(U)\otimes\CA$. As a shortcut, we shall write $a*b:=\tilde m_2(a,b)$. To rewrite these products in terms of the maps $\sft$ and $\acton$ of the crossed module which are independent of the form degree, we choose the convention of moving all form degrees to the left. Whenever two odd elements are moved past each other, a sign has to be inserted. For example, we have
\begin{equation}
\begin{aligned}
 A*B+B*A:=&\tilde m_2(A,B)+\tilde m_2(B,A)\\
 =&\dd x^\mu\wedge \dd x^\nu\wedge \dd x^\kappa(m_2(A_\mu,\tfrac12 B_{\nu\kappa})-m_2(\tfrac12 B_{\nu\kappa},A_\mu))\\
 =&\dd x^\mu\wedge \dd x^\nu\wedge \dd x^\kappa(A_\mu \acton \tfrac12 B_{\nu\kappa})\\
 =:&A\acton B~,
\end{aligned}
\end{equation}
where we used some coordinates $(x^\mu)$ on $U$ to illustrate the issue. 

\subsection{Higher flatness as an integrability condition}

The above constructions now suggest a higher generalization of the covariant derivative $\dd+A$ to the operator
\begin{equation}
 \nabla:=\tilde m_1+A*-B*~,
\end{equation}
where we inserted a sign for convenience. This operator will act on formal sums consisting of differential forms with values in $\sG$, $\frg$ and $\frh$. The detailed action is given in the following lemma.
\begin{lemma}\label{lem:higher_nabla}
 For $g\in\Omega^0(U)\otimes \sG$, $X\in \Omega^p(U)\otimes \frg$ and $Y\in\Omega^q(U)\otimes \frh$, we have the following two equations:
 \begin{equation*}
  \begin{aligned}
    \nabla (g+X+Y)&=\dd g+\dd X+\dd Y+(-1)^q\sft(Y)+Ag+AX+A*Y-B*g-B*X~,\\
    \nabla^2(g+X+Y)&=\CF g+\CF X+\CF*Y-H*g-H*X~.
  \end{aligned}
 \end{equation*} 
\end{lemma}
\begin{proof}
 The first equation follows directly. To compute the second equation, recall that $\tilde m_1$ satisfies by definition a Leibniz rule $\tilde m_1(a*b):=\tilde m_1(a)*b+(-1)^{\tilde{a}}a*m_1(b)$. We then have:
 \begin{equation}
 \begin{aligned}
  \nabla^2(g+X+Y)\ =\ &\nabla\big(\tilde m_1(g+X+Y)+A*(g+X+Y)-B*(g+X+Y)\big)\\
  =\ & \tilde m_1\big(A*(g+X+Y)-B*(g+X+Y)\big)+\\
  &+A*\tilde m_1(g+X+Y)-B*\tilde m_1(g+X+Y)+\\
  &+A*A*(g+X+Y)-A*B*(g+X+Y)+\\
  &-B*A(g+X)-B*A*Y-B*B*(g+X)\\
  =\ & \tilde m_1(A)*(g+X+Y)-\tilde m_1(B)*(g+X+Y)+\\
  &+A*A*(g+X+Y)-A*B*(g+X)-B*A(g+X)\\
  =\ & \CF(g+X)+\CF*Y-H*(g+X)~,
 \end{aligned}
 \end{equation}
 as claimed.
\end{proof}

We have now everything at our disposal to consider the higher analogue of the linear system \eqref{eq:lin-system} in the context of local connective structures on principal 2-bundles.
\begin{theorem}\label{thm:4.3}
 The equation 
 \begin{equation}\label{eq:lin-system-2}
 \nabla(g-\Lambda*g)=0~~~\mbox{for}~~~g\in \Omega^0(U)\otimes \sG~,~\Lambda\in\Omega^1(U)\otimes \frh
\end{equation}
implies that the local connective structure $(A,B)$ is pure gauge and that the curvature $\nabla^2=(\CF-H)$ vanishes.
\end{theorem}
\begin{proof}
 Using lemma \ref{lem:higher_nabla} with $X=0$ and $Y=-\Lambda*g$, we obtain
 \begin{equation}
  \nabla(g-\Lambda*g)=\dd g-(\dd \Lambda)*g-\Lambda*(\dd g)+\sft(\Lambda)g+Ag-A*\Lambda*g-B*g=0~.
 \end{equation}
 We can split this equation by form degree into
 \begin{equation}\label{eq:inbetween}
  \begin{aligned}
    0&=\dd g+Ag+\sft(\Lambda)g~,\\
    B*g&=(-\dd \Lambda-A\acton \Lambda)*g-\Lambda*(\dd g+Ag)~.
  \end{aligned}
 \end{equation}
 The first equation states that $A$ is pure gauge. Note that $\Lambda*\sft(\Lambda)=-\tfrac12[\Lambda,\Lambda]$, which is due to 
 \begin{equation}
  \begin{aligned}
   \Lambda*\sft(\Lambda)&=\tfrac12\big(\Lambda*\sft(\Lambda)+\Lambda*\sft(\Lambda)-\sft(\Lambda*\Lambda)\big)\\
   &=\tfrac12\big(\Lambda*\sft(\Lambda)-\sft(\Lambda)*\Lambda)=-\tfrac12 \sft(\Lambda)\acton \Lambda\\
   &=-\tfrac12 [\Lambda,\Lambda]~,
  \end{aligned}
 \end{equation}
 where we use the fact that $\sft$ is a derivation with respect to $m_2$ and the Peiffer identity. Using this identity together with $Y*g*g^{-1}=Y$, we can reformulate the second equation in \eqref{eq:inbetween} as 
 \begin{equation}
  B=-\dd \Lambda-A\acton \Lambda-\tfrac12[\Lambda,\Lambda]~,
 \end{equation}
 and the total local connective structure $(A,B)$ is pure gauge. A local connective structure which is pure gauge is clearly flat. Equivalently, lemma \ref{lem:higher_nabla} implies that $0=\nabla^2(g-\Lambda*g)=\CF g-\CF*\Lambda*g+H*g$ and therefore leads to the same conclusion.
\end{proof}

Altogether, we saw how the usual solution and integrability condition for the linear system \eqref{eq:lin-system} can be translated to the categorified case \eqref{eq:lin-system-2} by means of an associative 2-term $A_\infty$-algebra.

Finally, let us comment on the case of principal 3-bundles. Again, the extension of the discussion in the previous section to the case of local connective structures on principal 3-bundles is more or less a mere technicality. One starts from an associative 3-term $A_\infty$-algebra whose products extends to a 2-crossed module of matrix Lie groups. The covariant derivative is extended by adding a 3-form potential and the generalizations of the linear system \eqref{eq:lin-system-2} is rather straightforward. The same holds for the derivation of the analogous statements to theorem \ref{thm:4.3}

\section*{Acknowledgments}

We would like to thank Heiko Gimperlein for discussions. We are also very grateful to an anonymous referee whose extraordinarily detailed and very helpful comments led to many improvements. This work was supported by the EPSRC Career Acceleration Fellowship EP/H00243X/1.

\appendices
 
\subsection{Higher distributions leading to differential ideals}\label{app:HigherDistributions}

In this appendix, we briefly present a relation between certain higher distributions and differential ideals, generalizing the correspondence between ordinary involutive distributions and differential ideals generated by 1-forms. This is a first step towards a generalized Frobenius theorem.

Recall that a distribution $\CCD$ is a smoothly varying family of subspaces $\CCD_x$ of the fibers $T_xM$ of the tangent bundle of some manifold $M$. It is involutive if the Lie algebra of vector fields closes on sections of $\CCD$. That is, for any point $p\in M$, there is a neighborhood $U_p$ and vector fields $X_1,\ldots, X_r\in \frX(U_p)$ such that the $X_i$ are linearly independent and at each point $x\in U_p$, $\CCD_x$ is spanned by the $X_i$. Extending these vector fields to a local basis $X_1,\ldots,X_d$ of $TM$, we have 
\begin{equation}\label{eq:integrability}
 [X_i,X_j]=f^k_{ij} X_k\ewith f^{\underline k}_{\bar i\bar j}=0
\end{equation}
where the $f^k_{ij}$ are functions on $U_p$ and overlined and underlined indices $\bar{i}$ and $\underline{i}$ denote indices $i\leq r$ and $i>r$, respectively. 

Recall that by the Frobenius theorem, such an involutive distribution induces a regular foliation of the manifold $M$.

The Lie algebra of vector fields in \eqref{eq:integrability} has a dual Chevalley--Eilenberg algebra, which is encoded in the relations
\begin{equation}
 \dd \theta^k=-\tfrac12 f^k_{ij} \theta^i\wedge \theta^j~,
\end{equation}
where the 1-forms $\theta^i$ locally span $T^*M$ and satisfy $\theta^i(X_j)=\delta^i_j$. Note that because of $f^{\underline k}_{\bar i\bar j}=0$, the 1-forms $\theta^{\underline i}$ form a differential ideal.

This yields the modern formulation of the Frobenius theorem, which states that for a differential ideal on a manifold $M$ which is generated by 1-forms, there are submanifolds $e:N_p\embd M$ for each point $p\in M$ such that $p\in N_p$ and $e^* \alpha=0$ for any $\alpha $ in the differential ideal.

Let us now generalize the correspondence between certain distribution and differential ideals. We start by recalling some basic facts on multivector fields.

Consider a patch $U$ of a $d$-dimensional manifold $M$ together with the set of multivector fields $\frX^\bullet(U):=\Gamma(TU)\oplus \Gamma(\wedge^2TU)\oplus \cdots\oplus \Gamma(\wedge^d TU)$. On $\frX^\bullet(U)$, there is a natural generalization of the Lie bracket, which fulfills the Leibniz rule with respect to the $\wedge$-product: 
\begin{definition}
 The \uline{Schouten--Nijenhuis bracket} is the bilinear extension to $\frX^\bullet(U)$ of
\begin{equation}
\begin{aligned}
 &[V_1\wedge \cdots \wedge V_m,W_1\wedge \cdots \wedge W_n]_S:=\\
 &\hspace{1.5cm}\sum_{i,j=1}^{m,n}(-1)^{i+j}[V_i,W_j]\wedge V_1\wedge \cdots \wedge \hat{V}_i\wedge \cdots \wedge V_m\wedge W_1\wedge \cdots \wedge \hat{W}_j\wedge \cdots \wedge W_n~,
\end{aligned}
\end{equation}
where $V_i$, $W_j\in \frX^1(U)$ and $\hat{\cdot}$ indicates an omission.
\end{definition}

Note that the Schouten--Nijenhuis bracket turns the complex $\frX^\bullet(U)$ into a graded Lie algebra $L_0$. This graded Lie algebra has a dual Chevalley--Eilenberg algebra description in terms of forms in $\Omega^\bullet(U)$. Given a local basis $\theta^i,\xi^a,\ldots$ of linearly independent 1-forms, 2-forms, ..., spanning $T^*_xU$, $\wedge^2 T^*_xU$, $\ldots$ at every $x\in U$ we have
\begin{equation}
\dd \theta^i=-\tfrac12 f^i_{jk} \theta^j\wedge \theta^k~,~~~\dd \xi^a =-d^a_{ib}\theta^i\wedge \xi^b~,~~~\ldots~,
\end{equation}
where the $f^i_{jk}$ are the structure constants of the Lie algebra of vector fields and the additional structure constants $d^a_{ib}$ are functions on $U$ determined by the $f^i_{jk}$. As the $\theta^i,~\xi^a~,\ldots$ form a complete basis, we can also write these relations as
\begin{equation}\label{eq:Schouten-SHL}
\begin{aligned}
 \dd \theta^i&=-\tfrac12 \tilde{f}^i_{jk} \theta^j\wedge \theta^k+ \tilde{t}^i_a \xi^a~,\\
 \dd \xi^a&=-\tilde{d}^a_{ib} \theta^i\wedge \xi^b-\tfrac{1}{3!} \tilde{c}^a_{ijk}\theta^i\wedge \theta^j\wedge \theta^k~,\\
 \ldots
\end{aligned}
\end{equation}
where $\xi^a=m^a_{ij} \theta^i\wedge \theta^j$ and 
\begin{equation}
 f^i_{jk}=\tilde f^i_{jk}+\tilde t^i_a m^a_{jk}~,~~~ d^a_{ib}m^b_{jk}=\tilde d^a_{ib}m^b_{jk}+\tilde c^a_{ijk}~,~~~\ldots
\end{equation}
Equation \eqref{eq:Schouten-SHL} describes the Chevalley--Eilenberg algebra of a strong homotopy Lie algebra\footnote{See \cite{Lada:1992wc,0821843621} for a definition and more details.}. 
\begin{proposition}
The tilded structure constants in \eqref{eq:Schouten-SHL} define a strong homotopy Lie algebra on the graded vector space of multivector fields $\frX^\bullet(U)$.  
\end{proposition}
\noindent In particular, in terms of a basis $X_i\in \frX^1(U)$, $Y_a\in \frX^2(U)$, ... dual to that of $\Omega^\bullet(U)$ used above, we have the following higher brackets:
\begin{equation}\label{eq:associated_SHL}
\begin{aligned}
 \mu_1(Y_a)&=\tilde t_a^i X_i~,~~~&\mu_2(X_i,X_j)&=\tilde f^k_{ij} X_k~,\\
 \mu_2(X_i,Y_a)&=\tilde d^b_{ia} Y_b~,~~~&\mu_3(X_i,X_j,X_k)&=\tilde c^a_{ijk} Y_a~,\\
 \ldots
\end{aligned}
\end{equation}
The two underlying Chevalley--Eilenberg complexes of the Lie algebra $L_0$ given by the Schouten--Nijenhuis bracket and any $L_\infty$-algebra on $\frX^\bullet(U)$ given by a rewriting as in \eqref{eq:Schouten-SHL} are essentially identical. Therefore, there is an $L_\infty$-algebra isomorphisms between these, which motivates the following definition.
\begin{definition}
 An \uline{$L_\infty$-algebra associated to the Lie algebra $L_0$} is an $L_\infty$-algebra-structure on $\frX^\bullet(U)$ with higher brackets as in \eqref{eq:associated_SHL} obtained by a rewriting of the underlying Cheval\-ley--Eilenberg algebra of $L_0$ as in \eqref{eq:Schouten-SHL}. 
\end{definition}

Finally, note that we can truncate the structures introduced above from $\frX^\bullet(U)$ to multivector fields of a maximal degree $n$. In particular, we can evidently truncate the Schouten--Nijenhuis bracket to the complex 
\begin{equation}\label{eq:complex-Xi}
\frX_{(n)}(U)=
TU\ \longleftarrow\ \wedge^2 TU\ \longleftarrow\ \wedge^3 TU\ \longleftarrow\  \cdots \ \longleftarrow\ \wedge^n TU
\end{equation}
by setting
\begin{equation}
  [X_1\wedge \cdots \wedge X_p,Y_1\wedge \cdots \wedge Y_q]:=0
\end{equation}
for $X_i,Y_i\in \frX^1(U)$ and $p+q>n+1$. The associated $L_\infty$-algebras come then with higher brackets satisfying
\begin{equation}
  \mu_k(X_1,\ldots, X_k):=0
\end{equation}
for homogeneously graded $X_i\in \frX^{|X_i|}\subset \frX_{(n)}(U)$ and $k>n+1$ or $|X_1|+\cdots |X_k|>n+1$.

\

We now come to a generalization of the notion of distribution based on multivector fields.
\begin{definition}
 An \uline{$n$-distribution} on a $d$-dimensional manifold $M$ with $n\leq d$ is a sequence of distributions $\CCD=(\CCD_1,\ldots,\CCD_n)$ such that $\CCD_i$ is a distribution in $\wedge^i TM$.
\end{definition}
\noindent The notion of a pre-involutive distribution is now defined as follows:
\begin{definition}
 An $n$-distribution $\CCD$ on a manifold $M$ is called \uline{pre-involutive}, if there is an $L_\infty$-algebra associated to $L_0$, which closes on $\CCD$.
\end{definition}
\noindent In the case $n=1$, the above two definitions trivially reduce to those of an ordinary distribution and an ordinary involutive distribution.

In the following, let again $X_i\in \frX^1(U)$, $Y_a\in \frX^2(U)$, ... form a local basis spanning $TU$, $\wedge^2 TU$, ... and let $X_i$, $i\leq r_1$, $Y_a$, $a\leq r_2$, ... span a pre-involutive $n$-distribution $\CCD=(\CCD_1,\CCD_2,\ldots,\CCD_n)$. We shall again underline indices larger than $r_i$ and overline indices that are less or equal to $r_i$.  Using this notation, we can characterize the structure constants of $L_\infty$-algebras on pre-involutive $n$-distributions in more detail.
\begin{lemma}\label{lem:structure-constant-restrictions}
 The closure of an $L_\infty$-algebra associated to $L_0$ on a pre-involutive $n$-distri\-bution is equivalent to its structure constants $s^\alpha_{\beta_1\cdots\beta_k}=(\tilde t^i_a, \tilde f^k_{ij},\tilde d^b_{ia}, \tilde c^a_{ijk},\ldots)$ satisfying
 \begin{equation}
  s^{\underline \alpha}_{\overline \beta_1\cdots\overline \beta_k}=0~.
 \end{equation}
\end{lemma}

Let us now switch to the dual picture and consider the Chevalley--Eilenberg description of the above $n$-term $L_\infty$-algebra. That is, we have a local basis of forms $\theta^i\in \Omega^1(U)$, $\xi^a\in \Omega^2(U)$, ... with $\di_{X_i}\theta^j=\delta_i^j$, $\di_{Y_a}\xi^b=\delta_a^b$, etc. Closure of an associated $L_\infty$-algebra on a pre-involutive $n$-distribution amounts here to the following:
\begin{theorem}
 The forms $\theta^{\underline i},\xi^{\underline a},\ldots$ spanning the annihilators of the distributions contained in a pre-involutive $n$-distribution generate a differential ideal.
\end{theorem}
\begin{proof}
 The Chevalley--Eilenberg description of the $L_\infty$-algebra associated to $L_0$ is of the form
 \begin{equation}
  \dd \omega^\alpha=\sum_k s^\alpha_{\beta_1\cdots\beta_k} \omega^{\beta_1}\wedge \ldots \wedge \omega^{\beta_k}~
 \end{equation}
 for general forms $\omega^\alpha\in \Omega^1(U)\oplus \cdots\oplus\Omega^n(U)$. With Lemma \eqref{lem:structure-constant-restrictions}, we conclude that
 \begin{equation}
  \dd \omega^{\underline \alpha}=\sum_k s^{\underline \alpha}_{\underline \beta_1\beta_2\cdots\beta_k} \omega^{\underline \beta_1}\wedge \omega^{\beta_2}\wedge\ldots \wedge \omega^{\beta_k}~,
 \end{equation}
 which states that the $\omega^{\underline \alpha}$ generate a differential ideal.
\end{proof}

Note that in the case $n=1$, this is just the familiar statement that the annihilator of an integrable distribution spans a differential ideal.
 
\bibliography{bigone}

\begin{thebibliography}{10}

\bibitem{Baez:2004in}
J.~C.~Baez and U.~Schreiber,
{\em Higher gauge theory: 2-connections on 2-bundles,}
{\tt \href{http://www.arxiv.org/abs/hep-th/0412325}{hep-th/0412325}}.

\bibitem{Sati:0801.3480}
H.~Sati, U.~Schreiber, and J.~Stasheff,
{\em $L_\infty$-algebra connections and applications to String- and
  Chern-Simons $n$-transport,}
in: ``Quantum Field Theory,'' eds. B. Fauser, J. Tolksdorf and E. Zeidler, p.
  303, Birkhauser 2009
[{\tt \href{http://www.arxiv.org/abs/0801.3480}{0801.3480 [math.DG]}}].

\bibitem{Baez:2010ya}
J.~C.~Baez and J.~Huerta,
{\em {An invitation to higher gauge theory},}
\href{http://dx.doi.org/10.1007/s10714-010-1070-9}{Gen. Relativ. Gravit. {\bf
  43} (2011) 2335} [{\tt \href{http://www.arxiv.org/abs/1003.4485}{1003.4485
  [hep-th]}}].

\bibitem{Witten:1995zh}
E.~Witten,
{\em {Some comments on string dynamics},}
proceedings of ``Strings ‘95'', USC, 1995
[{\tt \href{http://www.arxiv.org/abs/hep-th/9507121}{hep-th/9507121}}].

\bibitem{Saemann:2012uq}
C.~Saemann and M.~Wolf,
{\em {Non-abelian tensor multiplet equations from twistor space},}
\href{http://dx.doi.org/10.1007/s00220-014-2022-0}{Commun. Math. Phys. {\bf
  328} (2014) 527} [{\tt \href{http://www.arxiv.org/abs/1205.3108}{1205.3108
  [hep-th]}}].

\bibitem{Saemann:2013pca}
C.~Saemann and M.~Wolf,
{\em {Six-dimensional superconformal field theories from principal 3-bundles
  over twistor space},}
\href{http://dx.doi.org/10.1007/s11005-014-0704-3}{Lett. Math. Phys. {\bf 104}
  (2014) 1147} [{\tt \href{http://www.arxiv.org/abs/1305.4870}{1305.4870
  [hep-th]}}].

\bibitem{Jurco:2014mva}
B.~Jurco, C.~Saemann, and M.~Wolf,
{\em {Semistrict higher gauge theory},}
\href{http://dx.doi.org/10.1007/JHEP04(2015)087}{JHEP {\bf 1504} (2015) 087}
  [{\tt \href{http://www.arxiv.org/abs/1403.7185}{1403.7185 [hep-th]}}].

\bibitem{jacobowitz1978}
H.~Jacobowitz,
{\em The Poincar{\'e} lemma for $\dd\omega =F(x,\omega)$,}
\href{http://projecteuclid.org/euclid.jdg/1214434604}{J. Diff. Geom. {\bf 13}
  (1978) 361}.

\bibitem{Voronov:0905.0287}
T.~Voronov,
{\em On a non-Abelian Poincar\'e lemma,}
\href{http://dx.doi.org/10.1090/S0002-9939-2011-11116-X}{Proc. AMS {\bf 140}
  (2012) 2855} [{\tt \href{http://www.arxiv.org/abs/0905.0287}{0905.0287
  [math.DG]}}].

\bibitem{Igusa:0912.0249}
K.~Igusa,
{\em Iterated integrals of superconnections,}
{\tt \href{http://www.arxiv.org/abs/0912.0249}{0912.0249 [math.AT]}}.

\bibitem{MR1083148}
R.~L.~Bryant, S.~S.~Chern, R.~B.~Gardner, H.~L.~Goldschmidt, and
  P.~A.~Griffiths,
{\em Exterior differential systems,}
Springer-Verlag, New York, 1991.

\bibitem{Breen:math0106083}
L.~Breen and W.~Messing,
{\em Differential geometry of gerbes,}
\href{http://dx.doi.org/10.1016/j.aim.2005.06.014}{Adv. Math. {\bf 198} (2005)
  732} [{\tt
  \href{http://www.arxiv.org/abs/math.AG/0106083}{math.AG/0106083}}].

\bibitem{Aschieri:2003mw}
P.~Aschieri, L.~Cantini, and B.~Jur\v{c}o,
{\em Nonabelian bundle gerbes, their differential geometry and gauge theory,}
\href{http://dx.doi.org/10.1007/s00220-004-1220-6}{Commun. Math. Phys. {\bf
  254} (2005) 367} [{\tt
  \href{http://www.arxiv.org/abs/hep-th/0312154}{hep-th/0312154}}].

\bibitem{Bartels:2004aa}
T.~Bartels,
{\em Higher gauge theory I: 2-Bundles,} PhD thesis, University of
  California-Riverside (2006)
[{\tt \href{http://www.arxiv.org/abs/math.CT/0410328}{math.CT/0410328}}].

\bibitem{Baez:0307200}
J.~C.~Baez and A.~D.~Lauda,
{\em Higher-dimensional algebra V: 2-groups,}
\href{http://www.kurims.kyoto-u.ac.jp/EMIS/journals/TAC/volumes/12/14/12-14.pdf}{Th.
  App. Cat. {\bf 12} (2004) 423} [{\tt
  \href{http://www.arxiv.org/abs/math.QA/0307200}{math.QA/0307200}}].

\bibitem{Martins:2009aa}
J.~F.~Martins and R.~Picken,
{\em The fundamental Gray 3-groupoid of a smooth manifold and local
  3-dimensional holonomy based on a 2-crossed module,}
\href{http://dx.doi.org/10.1016/j.difgeo.2010.10.002}{Diff. Geom. App. {\bf 29}
  (2011) 179} [{\tt \href{http://www.arxiv.org/abs/0907.2566}{0907.2566
  [math.CT]}}].

\bibitem{Jurco:2009px}
B.~Jur\v{c}o,
{\em {Nonabelian bundle 2-gerbes},}
\href{http://dx.doi.org/10.1142/S0219887811004963}{Int. J. Geom. Meth. Mod.
  Phys. {\bf 08} (2011)~49} [{\tt
  \href{http://www.arxiv.org/abs/0911.1552}{0911.1552 [math.DG]}}].

\bibitem{Conduche:1984:155}
D.~Conduch{\'e},
{\em Modules crois{\'e}s g{\'e}n{\'e}ralis{\'e}s de longueur 2,}
\href{http://dx.doi.org/10.1016/0022-4049(84)90034-3}{J. Pure Appl. Algebra
  {\bf 34} (1984) 155}.

\bibitem{Schreiber:0802.0663}
U.~Schreiber and K.~Waldorf,
{\em Smooth functors vs. differential forms,}
\href{http://dx.doi.org/10.4310/HHA.2011.v13.n1.a7}{Homology, Homotopy Appl.
  {\bf 13} (2011) 143} [{\tt
  \href{http://www.arxiv.org/abs/0802.0663}{0802.0663 [math.DG]}}].

\bibitem{Ward:1977ta}
R.~S.~Ward,
{\em {On selfdual gauge fields},}
\href{http://dx.doi.org/10.1016/0375-9601(77)90842-8}{Phys. Lett. A {\bf 61}
  (1977)~81}.

\bibitem{Stasheff:1963aa}
J.~D.~Stasheff,
{\em On the homotopy associativity of H-spaces, I.,}
\href{http://dx.doi.org/10.1090/S0002-9947-1963-99939-9}{Trans. Amer. Math.
  Soc. {\bf 108} (1963) 275}.

\bibitem{Stasheff:1963ab}
J.~D.~Stasheff,
{\em On the homotopy associativity of H-spaces, II.,}
\href{http://dx.doi.org/10.1090/S0002-9947-1963-0158400-5}{Trans. Amer. Math.
  Soc. {\bf 108} (1963) 293}.

\bibitem{Stasheff:1992bb}
J.~Stasheff,
{\em Differential graded Lie algebras, quasi-Hopf algebras and higher homotopy
  algebras,}
Quantum groups (Leningrad, 1990), Lecture Notes in Math., vol. 1510, Springer,
  Berlin, 1992, pp. 120–137.

\bibitem{Lada:1992wc}
T.~Lada and J.~Stasheff,
{\em {Introduction to sh Lie algebras for physicists},}
\href{http://dx.doi.org/10.1007/BF00671791}{Int. J. Theor. Phys. {\bf 32}
  (1993) 1087} [{\tt
  \href{http://www.arxiv.org/abs/hep-th/9209099}{hep-th/9209099}}].

\bibitem{0821843621}
M.~Markl, S.~Shnider, and J.~Stasheff,
{\em Operads in algebra, topology and physics,}
Mathematical Surveys and Monographs, American Mathematical Society, 2002.

\end{thebibliography}

\bibliographystyle{latexeu}

\end{document}